\documentclass{amc}
\usepackage{amsmath}

\def\currentvolume{X}
 \def\currentissue{X}
  \def\currentyear{200X}
   
    \def\ppages{X--XX}
     \def\DOI{10.3934/xx.xx.xx.xx}

\newtheorem{theorem}{Theorem}
\newtheorem{corollary}{Corollary}

\newtheorem{lemma}{Lemma}
\newtheorem{proposition}{Proposition}

\theoremstyle{definition}
\newtheorem{definition}{Definition}
\newtheorem{remark}{Remark}

\newcommand{\F} {\ensuremath{\mathbb{F}}}

\title[Automorphisms of self-dual extremal codes]
      {The automorphism group of a self-dual $[72,36,16]$ code does not contain ${\mathcal S}_3$, ${\mathcal A}_4$ or
$D_8$}

\author[Martino Borello, Francesca Dalla Volta and Gabriele Nebe]{}


\email{m.borello1@campus.unimib.it; francesca.dallavolta@unimib.it;
nebe@math.rwth-aachen.de}

\thanks{M. Borello and F. Dalla Volta are members of INdAM-GNSAGA,
Italy. F. Dalla Volta and G. Nebe were partially supported by
MIUR-Italy via PRIN ``Group theory and applications''}

\begin{document}
\maketitle

\centerline{\scshape Martino Borello }
\medskip
{\footnotesize
 \centerline{Dipartimento di Matematica e Applicazioni}
   \centerline{Universit\`{a} degli Studi di Milano Bicocca}
   \centerline{20125 Milan, Italy}
}

\medskip

\centerline{\scshape Francesca Dalla Volta}
\medskip
{\footnotesize
 \centerline{Dipartimento di Matematica e Applicazioni}
   \centerline{Universit\`{a} degli Studi di Milano Bicocca}
   \centerline{20125 Milan, Italy}
}

\medskip
 \centerline{\scshape Gabriele Nebe}
\medskip
{\footnotesize
 \centerline{ Lehrstuhl D f\"ur Mathematik}
   \centerline{RWTH Aachen University}
   \centerline{52056 Aachen, Germany}
}

\medskip

\centerline{(Communicated by Aim Sciences)}
\medskip

\begin{abstract}
A computer calculation with {\sc Magma} shows that there is no
extremal self-dual binary code $\mathcal{C}$ of length $72$ whose
automorphism group contains the symmetric group of degree $3$, the
alternating group of degree $4$ or the dihedral group of order $8$.
Combining this with the known results in the literature one obtains
that $\textnormal{Aut}(\mathcal{C})$ has order at most $5$ or is
isomorphic to the elementary abelian group of order $8$.
\end{abstract}

\section{Introduction}

Let $\mathcal{C}=\mathcal{C}^{\perp}\leq \F_2^n$ be a binary
\textit{self-dual} code of length $n$. Then the weight
$\textnormal{wt}(c):=|\{i \ | \ c_i=1\}|$ of every $c\in\mathcal{C}$
is even. When in particular
$\textnormal{wt}(\mathcal{C}):=\{\textnormal{wt}(c) \ | \
c\in\mathcal{C}\}\subseteq 4 \mathbb{Z}$, the code is called
\textit{doubly-even}. Using invariant theory, one may show
\cite{MSmindis} that the minimum weight
$d(\mathcal{C}):=\min(\textnormal{wt}(\mathcal{C}\setminus \{0\}))$
of a doubly-even self-dual code is at most $4+4\left\lfloor
\frac{n}{24} \right\rfloor$. Self-dual codes achieving this bound
are called \textit{extremal}. Extremal self-dual codes of length a
multiple of $24$ are particularly interesting for various reasons:
for example they are always doubly-even \cite{Rshad} and all their
codewords of a given nontrivial weight support $5$-designs
\cite{AMdes}. There are unique extremal self-dual codes of length
$24$ (the extended binary Golay code $\mathcal{G}_{24}$) and $48$
(the extended quadratic residue code $QR_{48}$) and both have a
fairly big automorphism group (namely
$\textnormal{Aut}(\mathcal{G}_{24})\cong M_{24}$ and
$\textnormal{Aut}(QR_{48})\cong \text{PSL}_2(47)$). The existence of
an extremal code of length $72$ is a long-standing open problem
\cite{S}. A series of papers investigates the automorphism group of
a putative extremal self-dual code of length $72$ excluding most of
the subgroups of $\mathcal{S}_{72}$. The most recent result is
contained in \cite{Baut6} where the first author excluded the
existence of automorphisms of order $6$.\\
In this paper we prove that neither $\mathcal{S} _3$ nor ${\mathcal
A}_4$ nor $D_8$ is contained in the automorphism group of such a
code.\\ The method to exclude ${\mathcal S}_3$ (which is isomorphic
to the dihedral group of order $6$)  is similar to that used for the
dihedral group of order $10$ in \cite{FeulNe} and based on the
classification of additive trace-Hermitian self-dual codes in $\F
_4^{12}$ obtained in \cite{DPadd}.\\ For the alternating group
${\mathcal A}_4$ of degree $4$ and the dihedral group $D_8$ of order
$8$, we use their structure as a semidirect product of an elementary
abelian group of order $4$ and a group of order $3$ and $2$
respectively. By \cite{Neven} we know that the fixed code of any
element of order $2$ is isomorphic to a
 self-dual binary code $D$ of length $36$ with minimum distance $8$.
These codes have been classified in \cite{Gaborit}; up to
equivalence there are $41$ such codes $D$. For all possible lifts
$\tilde{D} \leq \F_2^{72}$ that respect the given actions we compute
the codes ${\mathcal E} := \tilde{D}^{{\mathcal A}_4}$ and
${\mathcal E} := \tilde{D}^{D_8}$ respectively. We have respectively
only three and four such codes ${\mathcal E}$ with minimum distance
$\geq 16$. Running through all doubly-even ${\mathcal
A}_4$-invariant self-dual overcodes of ${\mathcal E}$ we see that no
such code is extremal. Since the group $D_8$ contains a cyclic group
of order $4$, say $C_4$, we use the fact \cite{Neven} that
$\mathcal{C}$ is a free $\F_2C_4$-module. Checking all doubly-even
self-dual overcodes of ${\mathcal E}$ which are free
$\F_2C_4$-modules we see that, also in this case, none is extremal.\\
The present state of research is summarized in the following
theorem.
\begin{theorem}
The automorphism group of a self-dual $[72,36,16]$ code is either
cyclic of order $1,2,3,4,5$ or elementary abelian of order $4$ or
$8$.
\end{theorem}
All results are obtained using extensive computations in {\sc Magma}
\cite{Magma}.

\section{The symmetric group of degree 3.}
\subsection{Preliminaries}\label{prel}
Let $\mathcal{C}$ be a binary self-dual code and let $g$ be an
automorphism of $\mathcal{C}$ of odd prime order $p$. Define
$\mathcal{C}(g):=\{c\in\mathcal{C} \ | \ c^g=c \}$ and
$\mathcal{E}(g)$ the set of all the codewords that have even weight
on the cycles of $g$. From a module theoretical point of view,
$\mathcal{C}$ is a $\F_2\langle g\rangle$-module and
$\mathcal{C}(g)=\mathcal{C} \cdot (1+g+\ldots+g^{p-1})$ and
$\mathcal{E}(g)=\mathcal{C} \cdot
(g+\ldots+g^{p-1})$. \\
In \cite{Hodd} Huffman notes (it is a special case of Maschke's
theorem) that
$$\mathcal{C} = \mathcal{C}(g)\oplus \mathcal{E}(g).$$
In particular it is easy to prove that the dimension of
$\mathcal{E}(g)$ is $\frac{(p-1)\cdot c}{2}$ where $c$ is the number
of cycles of $g$. In the usual manner we can identify vectors of
length $p$ with polynomials in $\mathcal{Q}:=\F_2[x]/(x^p-1)$; that
is $(v_1,v_2,\ldots,v_p)$ corresponds to $v_1+v_2x+\ldots+v_p
x^{p-1}$. The weight of a polynomial is the number of nonzero
coefficients. Let $\mathcal{P}\subset \mathcal{Q}$ be the set of all
even weight polynomials. If $1+x+\ldots+x^{p-1}$ is irreducible in
$\F_2[x]$ then $\mathcal{P}$ is a field with identity
$x+x^2+\ldots+x^{p-1}$ \cite{Hodd}. There is a natural map that we
will describe only in our particular case in the next section, from
$\mathcal{E}(g)$ to $\mathcal{P}^c$. Let us observe here only the
fact that, if $p=3$, then $1+x+x^2$ is irreducible in $\F_2[x]$ and
$\mathcal{P}$ is isomorphic to $\F_4$, the field with four elements.
The identification is the following:

{

\center{
\begin{tabular}{||c|c||c|c||}
  \hline
  0 & 000 & $\omega$ & 110 \\
  \hline
  1 & 011 & $\overline{\omega}$ & 101 \\
  \hline
\end{tabular}

} }

\subsection{The computations for ${\mathcal S}_3$}
Let $\mathcal{C}$ be an extremal self-dual code of length $72$ and
suppose that $G \leq \textnormal{Aut}(\mathcal{C})$ with $G\cong
\mathcal{S}_3$. Let $\sigma $ denote an element of order $2$ and $g$
an element of order $3$ in $G$. By \cite{Bord2} and \cite{Bord3},
$\sigma $ and $g$ have no fixed points. So, in particular, $\sigma $
has $36$ $2$-cycles and $g$ has $24$ $3$-cycles. Let us suppose,
w.l.o.g. that
$$\sigma =(1,4)(2,6)(3,5)\ldots(67,70)(68,72)(69,71)$$
and
$$g=(1,2,3)(4,5,6)\ldots(67,68,69)(70,71,72).$$
As we have seen in Section \ref{prel},
$$\mathcal{C}=\mathcal{C}(g)\oplus \mathcal{E}(g)$$
where $\mathcal{E}(g)$ is the subcode of $\mathcal{C}$ of all the
codewords with an even weight on the cycles of $g$, of dimension
$24$. We can consider a map
$$f:\mathcal{E}(g)\rightarrow \F_4^{24}$$
extending the identification $\mathcal{P}\cong \F_4$, stated in
Section \ref{prel}, to each cycle of $g$. \\
Again by \cite{Hodd},  $\mathcal{E}(g)':=f(\mathcal{E}(g))$ is an
Hermitian self-dual code over $\F_4$ (that is
$\mathcal{E}(g)'=\left\{\epsilon\in\F_4^{24} \ \left| \
\sum_{i=0}^{24} \epsilon_i \overline{\gamma_i} =0 \ \text{for all} \
\gamma \in \mathcal{E}(g)' \right.\right\}$, where
$\overline{\alpha}=\alpha ^2$ is the conjugate of $\alpha$ in
$\F_4$). Clearly the minimum distance of $\mathcal{E}(g)'$ is
$\geq8$. So $\mathcal{E}(g)'$ is a $[24,12,\geq 8]_4$ Hermitian
self-dual code. \\
The action of $\sigma $ on $\mathcal{C}\leq \F_2^{72}$ induces an
action on $\mathcal{E}(g)'\leq \F_4^{24}$, namely
$$(\epsilon_1,\epsilon_2,\ldots,\epsilon_{23},\epsilon_{24})^{\sigma }=
(\overline{\epsilon_2},\overline{\epsilon_1},\ldots,\overline{\epsilon_{24}},
\overline{\epsilon_{23}})$$ Note that this action is only
$\F_2$-linear. In particular, the subcode fixed by $\sigma $, say
$\mathcal{E}(g)'(\sigma )$, is
$$\mathcal{E}(g)'(\sigma )=\{(\epsilon_1,\overline{\epsilon_1},\ldots,\epsilon_{12},\overline{\epsilon_{12}})\in \mathcal{E}(g)'\}$$

\begin{proposition} \textnormal{(cf. \cite[Cor. 5.6]{FeulNe})}
The code
$$\mathcal{X}:=\pi(\mathcal{E}(g)'(\sigma )):=\{(\epsilon_1,\ldots,\epsilon_{12}) \in \F_4^{12}
\ | \
(\epsilon_1,\overline{\epsilon_1},\ldots,\epsilon_{12},\overline{\epsilon_{12}})\in
\mathcal{E}(g)'\}$$ is an additive trace-Hermitian self-dual
$(12,2^{12},\geq 4)_4$ code such that
$$\mathcal{E}(g)':=\phi(\mathcal{X}):=\langle(\epsilon_1,\overline{\epsilon_1},
\ldots,\epsilon_{12},\overline{\epsilon_{12}}) \ | \ (\epsilon_1,
\ldots,\epsilon_{12})\in\mathcal{X}\rangle_{\F_4}. $$
\end{proposition}

\begin{proof}
For $\gamma , \epsilon \in \mathcal{X} $  the inner product of their
preimages in $\mathcal{E}(g)'(\sigma )$ is
$$ \sum_{i=1}^{12}
(\epsilon_i\overline{\gamma_i}+\overline{\epsilon_i}\gamma_i)  $$
which is $0$ since $\mathcal{E}(g)'(\sigma )$ is self-orthogonal.
Therefore $\mathcal{X}$ is trace-Hermitian self-orthogonal. Thus
$$\dim_{\F_2}(\mathcal{X}) = \dim _{\F_2} (
\mathcal{E}(g)'(\sigma ) ) = \frac{1}{2} \dim _{\F_2}
(\mathcal{E}(g)' ) $$ since $\mathcal{E} (g)'$ is a projective
$\F_2\langle\sigma \rangle$-module, and so $\mathcal{X}$ is
self-dual. Since $\dim_{\F_2}(\mathcal{X})=12=\dim_{\F_{4}}
(\mathcal{E}(g)')$, the $\F_4$-linear code $\mathcal{E}(g)'\leq
\F_4^{24}$ is obtained from $\mathcal{X}$ as stated.
\end{proof}

 All additive trace-Hermitian self-dual
codes in $\F_4^{12}$ are classified in \cite{DPadd}. There are
$195,520$ such codes that have minimum distance $\geq 4$ up to
monomial equivalence.

\begin{remark}
If $\mathcal{X}$ and $\mathcal{Y}$ are monomial equivalent, via a
$12\times 12$ monomial matrix $M:=(m_{i,j})$, then
$\phi(\mathcal{X})$ and $\phi(\mathcal{Y})$ are monomial equivalent
too, via the $24\times 24$ monomial matrix $M':=(m'_{i,j})$, where
$m'_{2i-1,2j-1}=m_{i,j}$ and $m'_{2i,2j}=\overline{m_{i,j}}$, for
all $i,j\in\{1,\ldots,12\}$.
\end{remark}

An exhaustive search with {\sc Magma} (of about $7$ minutes CPU on
an Intel(R) Xeon(R) CPU X5460 @ 3.16GHz) shows that the minimum
distance of $\phi(\mathcal{X})$ is $\leq 6$, for each of the
$195,520$ additive trace-Hermitian self-dual $(12,2^{12},\geq 4)_4$
codes. But $\mathcal{E}(g)'$ should have minimum distance $\geq 8$,
a contradiction. So we proved the following.

\begin{theorem}
The automorphism group of a self-dual $[72,36,16]$ code does not
contain a subgroup isomorphic to   $\mathcal{S}_3$.
\end{theorem}

\section{The alternating group of degree 4 and the dihedral group of order 8.}

\subsection{The action of the Klein four group.}

For the alternating group ${\mathcal A}_4$ of degree 4 and the
dihedral group $D_8$ of order 8 we use their structure
$$\begin{array}{ccc}
{\mathcal A}_4 \cong & {\mathcal V}_4 : C_3 \cong  & (C_2\times C_2
):C_3 =\langle g,h \rangle : \langle \sigma \rangle
\\
D_8 \cong & {\mathcal V}_4 : C_2 \cong  & (C_2\times C_2 ):C_2
=\langle g,h \rangle : \langle \sigma \rangle
\end{array}
$$
as a semidirect product.

Let ${\mathcal C}$ be some extremal $[72,36,16]$ code such that
${\mathcal H} \leq \textnormal{Aut}({\mathcal C})$ where ${\mathcal
H}\cong \mathcal{A}_4$ or ${\mathcal H}\cong D_8$. Then by
\cite{Bord2} and \cite{Bord3} all non trivial elements in ${\mathcal
H}$ act without fixed points and we may replace ${\mathcal C}$ by
some equivalent code so that
$$\begin{array}{llr}
g = & (1, 2)(3, 4)(5, 6)(7, 8)(9, 10)(11, 12) \ldots (71, 72) &  \\
h = & (1, 3)(2, 4)(5, 7)(6, 8)(9, 11)(10, 12) \ldots (70,72)  &  \\
\sigma = & (1, 5, 9)(2, 7, 12)(3, 8, 10)(4, 6, 11) \ldots (64, 66, 71) & (for \ {\mathcal A}_4) \\
\sigma = & (1, 5)(2, 8)(3, 7)(4, 6) \ldots (68, 70)  & (for \ D_8) \\
\end{array}
$$

Let
$${\mathcal G}:= C_{{\mathcal S}_{72}}({\mathcal H}) := \{ t\in {\mathcal
S}_{72} \mid tg=gt, th=ht, t\sigma = \sigma t \} $$ denote the
centralizer of this subgroup ${\mathcal H}$ in ${\mathcal S}_{72}$.
Then ${\mathcal G}$ acts on the set of extremal ${\mathcal
H}$-invariant self-dual codes and we aim to find a system of orbit
representatives for this action.

\begin{definition}
Let
$$\begin{array}{l} \pi _1:  \{ v\in \F _2^{72} \mid v^g = v \} \to \F_2^{36}  \\
(v_1,v_1,v_2,v_2, \ldots,v_{36},v_{36} ) \mapsto
 (v_1,v_2,\ldots , v_{36} ) \end{array}  $$
denote the bijection between the fixed space of $g$ and $\F_2^{36}$
and
$$\begin{array}{l}
\pi _2: \{ v \in \F_2^{72} \mid v^g = v \mbox{ and } v^h = v \} \to \F_2^{18} \\
(v_1,v_1,v_1,v_1,v_2,\ldots , v_{18}) \mapsto (v_1,v_2,\ldots ,
v_{18}) \end{array} $$ the bijection between the fixed space of
$\langle g,h \rangle \triangleleft {\mathcal A}_4$ and $\F_2^{18} $.
Then $h$ acts on the image of $\F_2^{18}$ as
$$ (1,2)(3,4)\ldots  (35,36) . $$
Let
$$\begin{array}{l} \pi _3 : \{ v\in \F _2^{36} \mid v^{\pi_1(h)} = v \} \to \F_2^{18},  \\
(v_1,v_1,v_2,v_2, \ldots , v_{18},v_{18} ) \mapsto
 (v_1,v_2,\ldots , v_{18} ) , \end{array} $$ so that $\pi _2 = \pi _3 \circ \pi _1 $.
\end{definition}

\begin{remark}
The centraliser $C_{\mathcal{S}_{72}}(g) \cong C_2 \wr
\mathcal{S}_{36} $ of $g$ acts on the set of  fixed points of $g$.
Using the isomorphism $\pi _1$ we obtain a group epimorphism which
we again denote by $\pi _1$
$$\pi _1 : C_{\mathcal{S}_{72}}(g) \to \mathcal{S}_{36} $$
with kernel $C_2^{36}$. Similarly we obtain the epimorphism
$$\pi_3:C_{\mathcal{S}_{36}}(\pi_1(h))\rightarrow
\mathcal{S}_{18} .$$ The normalizer $N_{\mathcal{S}_{72}}(\langle
g,h\rangle) $ acts on the set of $\langle g,h\rangle$-orbits which
defines a homomorphism
$$\pi_2:N_{\mathcal{S}_{72}}(\langle g,h\rangle)\rightarrow \mathcal{S}_{18}.$$
\end{remark}

Let us consider the fixed code $ {\mathcal C}(g) $ which is
isomorphic to
 $$\pi_1({\mathcal C}(g))  = \{ (c_1,c_2,\ldots , c_{36} ) \mid
(c_1,c_1,c_2,c_2, \ldots c_{36},c_{36} ) \in {\mathcal C} \} .$$ By
\cite{Neven}, the code $\pi_1({\mathcal C}(g))  $ is some self-dual
code of length $36$ and minimum distance $8$. These codes have been
classified in \cite{Gaborit}; up to equivalence
 (under the action of
the full symmetric group $\mathcal{S}_{36}$) there are $41$ such
codes. Let $$Y_1,\ldots , Y_{41}$$ be a system of representatives of
these extremal self-dual codes of length $36$.

\begin{remark}
$\mathcal{C}(g)\in\mathcal{D} $ where
$${\mathcal D} := \left\{D \leq \F_2^{36}  \left|  \begin{array}{c} D=D^\perp, d(D)=8,
\pi_1(h) \in \textnormal{Aut}(D) \\ \mbox{ and } \pi_2(\sigma ) \in
\textnormal{Aut} (\pi_3(D(\pi_1(h)))) \end{array} \right.\right\}
.$$ For $1\leq k\leq 41$ let ${\mathcal D}_k := \{ D\in {\mathcal D}
\mid D \cong Y_k \}$.
\end{remark}

Let ${\mathcal G}_{36} := \{ \tau \in C_{{\mathcal S}_{36}} (\pi
_1(h) ) \mid \pi_3 (\tau ) \pi _2(\sigma ) = \pi_2(\sigma ) \pi
_3(\tau ) \} $.

\begin{remark}
For ${\mathcal H} \cong {\mathcal A}_4$ the group ${\mathcal
G}_{36}$ is isomorphic to $ C_2\wr C_3\wr  {\mathcal S}_6 $. It
contains $\pi _1({\mathcal G}) \cong {\mathcal A}_4 \wr {\mathcal
S}_6$ of index $64$.
\\
For ${\mathcal H}\cong D_8$ we get ${\mathcal G}_{36} =
\pi_1({\mathcal G}) \cong C_2\wr C_2 \wr {\mathcal S}_9$.
\end{remark}

\begin{lemma}\label{repr}
A set of representatives of the ${\mathcal G}_{36}$ orbits on
${\mathcal D}_k$ can be computed by performing the following
computations:
\begin{itemize}
\item Let $h_1,\ldots , h_s$ represent the
conjugacy classes  of fixed point free elements of order $2$ in
$\textnormal{Aut}(Y_k)$.
\item Compute elements $\tau _1,\ldots, \tau  _s \in {\mathcal S}_{36}$ such that
$\tau  _i^{-1} h_i \tau  _i =\pi _1(h) $ and put $D_i := Y_k^{\tau
_i}$ so that
 $\pi_1(h) \in \textnormal{Aut} (D_i)$.
\item For all $D_i$
let $\sigma _1,\ldots , \sigma _{t_i}$ a set of representives of the
action by conjugation  by the subgroup
$\pi_3(C_{\textnormal{Aut}(D_i)}(\pi_1(h)))$ on fixed point free
elements of order $3$ (for ${\mathcal H} \cong {\mathcal A}_4$)
respectively $2$ (for ${\mathcal H} \cong D_8 $)
 in $\textnormal{Aut} (\pi _3 (D_i(\pi_1(h)) ) ) $.
\item  Compute elements $\rho _1,\ldots \rho  _{t_i} \in {\mathcal S}_{18}$ such that
$\rho  _j^{-1} \sigma_j \rho  _j =\pi _3(\sigma) $, lift $\rho _j$
naturally to a permutation $\tilde{\rho }_j \in {\mathcal S}_{36}$
commuting with $\pi _1 (h) $ (defined by $\tilde{\rho } _j (2a-1) =
2 \rho_j(a) -1 $,
 $\tilde{\rho _j} (2a) = 2 \rho_j(a)  $)
and put $$D_{i,j} := (D_i)^{\tilde{\rho } _j} = Y_k^{\tau _i
\tilde{\rho }_j} $$ so that
 $\pi_3(\sigma ) \in \textnormal{Aut}(\pi _2(D_{i,j}(\pi_1(h))))$.
\end{itemize}
Then $\{ D_{i,j} \mid 1\leq i \leq s, 1\leq j \leq t_i \} $
represent the ${\mathcal G}_{36}$-orbits on ${\mathcal D}_k$.
\end{lemma}

\begin{proof}
Clearly these codes lie in ${\mathcal D}_k$. \\
Now assume that there is some $\tau \in {\mathcal G}_{36}$ such that
$$Y_k^{\tau _{i'} \tilde{\rho }_{j'} \tau } = D_{i',j'} ^{\tau } = D_{i,j}
= Y_k^{\tau _{i} \tilde{\rho }_{j}  }.$$ Then
$$\epsilon :=  \tau _{i'} \tilde{\rho }_{j'} \tau \tilde{\rho }_{j} ^{-1} \tau _{i} ^{-1}
\in \textnormal{Aut} (Y_k) $$ satisfies $\epsilon h_i \epsilon ^{-1}
= h_{i'} $, so
 $h_i$ and $h_{i'}$ are conjugate in $\textnormal{Aut} (Y_k)$, which implies $i=i'$
 (and so $\tau_i=\tau_{i'}$). Now,
$$Y_k^{\tau _{i} \tilde{\rho }_{j'} \tau } =D_{i}^{\tilde{\rho }_{j'}
\tau  }= D_i^{\tilde{\rho }_{j}  } = Y_k^{\tau _{i} \tilde{\rho
}_{j} }.$$ Then
$$\epsilon' :=  \tilde{\rho }_{j'} \tau \tilde{\rho }_{j} ^{-1}
\in \textnormal{Aut} (D_i) $$ commutes with $\pi _1(h)$. We compute
that $\pi_3(\epsilon') \sigma_j \pi_3({\epsilon'}^{-1}) =
\sigma_{j'} $ and hence  $j=j'$.

Now let $D \in {\mathcal D}_k$ and choose some $\xi \in {\mathcal
S}_{36}$ such that $D^{\xi } = Y_k $. Then $\pi_1(h) ^{\xi } $ is
conjugate to some of the chosen representatives $h_i \in
\textnormal{Aut}(Y_k)$ ($i=1,\ldots ,s$) and we may multiply $\xi $
by some automorphism of $Y_k$ so that $\pi _1(h) ^{\xi } = h_i =
\pi_1(h)^{\tau _i^{-1}} $. So $\xi \tau_i \in C_{{\mathcal S}_{36}}
(\pi_1(h))$ and $D ^{\xi \tau _i } = Y_k ^{\tau _i}=D_i $. Since
$\pi_3(\sigma ) \in \textnormal{Aut}(\pi_3(D(\pi _1(h))) ) $ we get
$$\pi _3(\sigma )^{\pi _3(\xi \tau _i)} \in \textnormal{Aut} (\pi _3(D_i (\pi _1(h)) ))  $$
and so there is some automorphism $\alpha \in
\pi_3(C_{\textnormal{Aut}(D_i)}(\pi_1(h)))$ and some $j \in
\{1,\ldots , t_i \}$ such that $(\pi _3(\sigma )^{\pi _3(\xi \tau
_i)} ) ^{\alpha } = \sigma _j$. Then
$$D^{\xi \tau_i \tilde{\alpha } \tilde{\rho }_j}=D_{i,j} $$ where
$\xi \tau_i \tilde{\alpha } \tilde{\rho }_j \in {\mathcal G}_{36}$.
\end{proof}

\subsection{The computations for ${\mathcal A}_4$.}

We now deal with the case ${\mathcal H}\cong {\mathcal A}_4$.

\begin{remark}\label{computedCodes}
With {\sc Magma} we use the algorithm given in Lemma \ref{repr} to
compute  that there are exactly $25,299$ ${\mathcal G}_{36}$-orbits
on ${\mathcal D}$, represented by, say, $X_1,\ldots , X_{25,299}$.
\end{remark}

As ${\mathcal G}$ is the centraliser of ${\mathcal A}_4$ in
${\mathcal S}_{72}$ the image $\pi _1({\mathcal G})$ commutes with
$\pi_1(h)$ and $\pi _2({\mathcal G})$ centralizes $\pi_2(\sigma )$.
In particular the group ${\mathcal G}_{36}$ contains $\pi
_1({\mathcal G})$ as a subgroup. With  {\sc Magma} we compute that
$[{\mathcal G}_{36} : \pi _1({\mathcal G})] = 64$. Let $g_1,\ldots ,
g_{64} \in {\mathcal G}_{36} $ be a left transversal of $ \pi
_1({\mathcal G})$ in ${\mathcal G}_{36} $ .

\begin{remark}
The set $\{ X_i ^{g_j} \mid 1\leq i\leq 25,299, 1\leq j \leq 64 \} $
contains a set of representatives of the $\pi _1({\mathcal
G})$-orbits on ${\mathcal D}$.
\end{remark}

\begin{remark}
For all $1\leq i \leq 25,299, 1\leq j\leq 64 $ we compute the code
$$ {\mathcal E} := E(X_i^{g_j},\sigma ) := \tilde{D} + \tilde{D}^{\sigma } + \tilde{D}^{\sigma ^2} ,
\mbox{ where } \tilde{D} = \pi_1^{-1} (X_i^{g_j}). $$ For three
$X_i$ there are two codes $\tilde{D}_{i,1} = \pi _1^{-1} (
X_i^{g_{j_1}}) $ and $\tilde{D}_{i,2} = \pi _1^{-1} ( X_i^{g_{j_2}})
$ such that $E(X_i^{g_{j_1}},\sigma )$ and $E(X_i^{g_{j_2}},\sigma
)$ are doubly even and of minimum distance $16$. In all three cases,
the two codes are equivalent. Let us call the inequivalent codes
${\mathcal E}_1, {\mathcal E}_2$ and ${\mathcal E}_3$, respectively.
They have dimension $26$, $26$, and $25$, respectively, minimum
distance $16$ and their automorphism groups are
$$\textnormal{Aut}({\mathcal E}_1) \cong {\mathcal S}_4,
\textnormal{Aut}({\mathcal E}_2) \mbox{ of order } 432,
\textnormal{Aut}({\mathcal E}_3) \cong ({\mathcal A}_4\times
{\mathcal A}_5):2.$$ All three groups contain a unique conjugacy
class of subgroups conjugate in ${\mathcal S}_{72}$ to ${\mathcal
A}_4$ (which is normal for ${\mathcal E}_1$ and ${\mathcal E}_3$).
\end{remark}

These computations took about $26$ hours CPU, using an Intel(R)
Xeon(R) CPU X5460 @ 3.16GHz.

\begin{corollary}
The code ${\mathcal C}(g) + {\mathcal C}(h) + {\mathcal C}(gh) $ is
equivalent under the action of ${\mathcal G}$ to one of the three
codes ${\mathcal E}_1, {\mathcal E}_2$ or ${\mathcal E}_3$.
\end{corollary}

Let ${\mathcal E}$ be one of these three codes. The group ${\mathcal
A}_4$ acts on ${\mathcal V}:={\mathcal E}^{\perp } / {\mathcal E} $
with kernel $\langle g,h \rangle $. The space ${\mathcal V}$ is
hence an $\F_2 \langle \sigma \rangle $-module supporting a $\sigma
$-invariant form such that ${\mathcal C} $ is a self-dual submodule
of ${\mathcal V}$. As in Section \ref{prel} we obtain a canonical
decomposition
$${\mathcal V} =  {\mathcal V}(\sigma ) \perp {\mathcal W} $$
where ${\mathcal V}(\sigma )$  is the fixed space of $\sigma $ and
$\sigma $ acts as a primitive third root of unity on ${\mathcal W}$.

For ${\mathcal E} = {\mathcal E}_1$ or ${\mathcal E} = {\mathcal
E}_2$ we compute that ${\mathcal V}(\sigma ) \cong \F_2^4$ and
${\mathcal W} \cong \F_4^{8}$. For both codes the full preimage of
every self-dual submodule of
 ${\mathcal V}(\sigma )$ is a code of minimum distance $<16$.

For ${\mathcal E} = {\mathcal E}_3$ the dimension of ${\mathcal
V}(\sigma )$ is $2$ and there is a unique self-dual submodule of
${\mathcal V}(\sigma )$ so that the full preimage $E_3$ is
doubly-even and of minimum distance $\geq 16$. The element $\sigma $
acts on $E_3^{\perp}/E_3 \cong {\mathcal W}$ with irreducible
minimal polynomial, so $E_3^{\perp } / E_3 \cong \F_4^{10}$. The
code ${\mathcal C}$ is a preimage of one of the $58,963,707$ maximal
isotropic $\F_4$-subspaces of the Hermitian $\F_4$-space $E_3^{\perp
}/ E_3$.

The unitary group $GU(10,2)$ of $E_3^{\perp}/E_3 \cong \F_4^{10}$
acts transitively on the maximal isotropic subspaces. So a quite convenient way to enumerate
all these spaces is to compute an isometry of $E_3^{\perp}/E_3$
with the standard model used in {\sc Magma} and then compute the
$GU(10,2)$-orbit of one maximal isotropic space (e.g. the one spanned by
the first 5 basis vectors in the standard model).
The problem here is that the orbit becomes too long to be stored in the
available memory (4GB). So
we first
compute all $142,855$ one dimensional isotropic subspaces
$\overline{E}_3/{E}_3 \leq _{\F_4} {E}_3^{\perp }/{E}_3 $ for which
the code $\overline{E}_3 $ has minimum distance $\geq 16$. The
automorphism group $\textnormal{Aut}(E_3) =
\textnormal{Aut}({\mathcal E}_3)$ acts on these codes with  $1,264$
orbits. For all these $1,264$ orbit representatives
${\overline{E}}_3$ we compute the $114,939$ maximal isotropic
subspaces of $\overline{E}_3^{\perp }/\overline{E}_3 $ (as the
orbits of one given subspace under the unitary group $GU(8,2)$ in
{\sc Magma}) and check whether the corresponding doubly-even
self-dual code has minimum distance $16$. No such code is found.

Note that the latter computation can be parallelised easily as all
$1,264$ computations are independent of each other. We split it into
$10$ jobs. To deal with $120$ representatives $\overline{E}_3$ took
between $5$ and $10$ hours on a Core i7 870  (2.93GHz) personal
computer.

This computation shows the following.

\begin{theorem}
The automorphism group of a self-dual $[72,36,16]$ code does not
contain  a subgroup isomorphic to $\mathcal{A}_4$.
\end{theorem}

\subsection{The computations for $D_8$.}

For this section we assume that ${\mathcal H}\cong D_8$. Then $\pi
_1({\mathcal G}) = {\mathcal G}_{36}$ and we may use Lemma
\ref{repr} to compute a system of representatives of the
$\pi_1({\mathcal G}) -$orbits on the set ${\mathcal D}$.

\begin{remark}\label{computedCodesD8}
$\pi_1({\mathcal G}) $ acts on ${\mathcal D}$ with exactly $9,590$
orbits represented by, say, $X_1,\ldots ,$ $X_{9,590}$. For all
$1\leq i \leq 9,590 $ we compute the code
$$ {\mathcal E} := E(X_i,\sigma ) := \tilde{D} + \tilde{D}^{\sigma },
\mbox{ where } \tilde{D} = \pi_1^{-1} (X_i). $$ For four $X_i$ the
code $E(X_i,\sigma )$ is doubly even and of minimum distance $16$.
Let us call the inequivalent codes ${\mathcal E}_1, {\mathcal E}_2,
{\mathcal E}_3$ and ${\mathcal E}_4$, respectively. All have
dimension $26$ and minimum distance $16$.
\end{remark}

\begin{corollary}
The code ${\mathcal C}(g) + {\mathcal C}(h) + {\mathcal C}(gh) $ is
equivalent under the action of ${\mathcal G}$ to one of the four
codes ${\mathcal E}_1, {\mathcal E}_2, {\mathcal E}_3$ or ${\mathcal
E}_4$.
\end{corollary}

This computation is very fast (it is due mainly to the fact that
$\mathcal{G}_{36}=\pi(\mathcal{G})$). It took about $5$ minutes CPU
on an Intel(R) Xeon(R) CPU X5460 @ 3.16GHz.

As it seems to be quite hard to compute all $D_8$-invariant
self-dual overcodes of ${\mathcal E}_i$ for these four codes
${\mathcal E}_i$ we apply a different strategy which is based on the
fact that $h = (g\sigma)^2$ is the square of an element of order
$4$. So let $$k:= g\sigma = (1, 8, 3, 6)(2, 5, 4, 7) \ldots (66, 69,
68, 71)  \in D_8 .$$ By \cite{Neven}, $\mathcal{C}$ is a free
$\F_2\langle k\rangle$-module (of rank $9$). Since $\langle
k\rangle$ is abelian, the module is both left and right; here we use
the right notation. The regular module $\F_2\langle k\rangle$ has a
unique irreducible module, $1$-dimensional, called the socle, that
is $\langle (1+k+k^2+k^3)\rangle$. So $\mathcal{C}$, as a free
$\F_2\langle k\rangle$-module, has socle
$\mathcal{C}(k)=\mathcal{C}\cdot(1+k+k^2+k^3)$. This implies that,
for every basis $b_1,\ldots,b_9$ of $\mathcal{C}(k)$, there exist
$w_1,\ldots,w_9\in {\mathcal C}$ such that
$w_i\cdot(1+k+k^2+k^3)=b_i$ and
$$\mathcal{C}=w_1\cdot \F_2\langle k\rangle \oplus \ldots \oplus w_9 \cdot
\F_2\langle k\rangle.$$
 To get all the possible
overcodes of $\mathcal{E}_i$, we choose a basis of the socle
$\mathcal{E}_i(k)$, say $b_1,\ldots,b_9$, and look at the sets
$$W_{i,j}=\{w+\mathcal{E}_i \in \mathcal{E}_i^\perp/\mathcal{E}_i \ | \ w\cdot(1+k+k^2+k^3)=b_j \ \text{and} \ d(\mathcal{E}_i+w\cdot \F_2\langle
k\rangle)\ge 16\}$$

For every $i$ we have at least one $j$ for which the set $W_{i,j}$
is empty. This computation (of about $4$ minutes CPU on the same
computer) shows the following.

\begin{theorem}
The automorphism group of a self-dual $[72,36,16]$ code does not
contain  a subgroup isomorphic to $D_8$.
\end{theorem}

\section*{Acknowledgment}

The authors like to express their gratitude to A. Previtali for the
fruitful discussions in Milan. They thank \emph{Laboratorio di
Matematica Industriale e Crittografia} of Trento for providing the
computational infrastructure including the computer algebra system {\sc Magma}.

\medskip
Received xxxx 20xx; revised xxxx 20xx.
\medskip

\end{document}